\documentclass[11pt,reqno]{article}
\usepackage{amsmath}
\usepackage{amsthm}
\usepackage{amssymb}
\usepackage{epsfig}
\usepackage{fancyheadings}
\usepackage{amsbsy, graphicx, amsmath, amsfonts}
\usepackage{natbib}
\usepackage{array,epsfig,amssymb}

\usepackage{calc}
\usepackage{xspace}
\usepackage{comment}
\usepackage{paralist}
\usepackage{url}
\usepackage{setspace}
\usepackage{kotex}
\usepackage{rotating}
\usepackage{array,arydshln}	
\usepackage{tabularx}	
\usepackage{multirow}
\usepackage{booktabs}
\usepackage{threeparttable}
\usepackage{color}	
\usepackage{authblk}
\usepackage{hyperref}
\usepackage{enumitem}
\usepackage{xspace}

\DeclareMathOperator*{\argmin}{argmin}

\newtheorem{lemma}{Lemma}
\newtheorem{theorem}{Theorem}

\newtheorem{remark}{Remark}

\newcommand {\A} {\alpha}
\newcommand {\E} {E}
\newcommand {\T} {\theta}
\newcommand{\HT} {\hat{\theta}_{\A,n}}
\newcommand {\ep} {\epsilon}
\newcommand {\pa} {\partial}
\newcommand {\paa}{\partial^2_\T}

\title{ A robust approach for testing parameter change\\ in Poisson autoregressive models}
\author[1]{Jiwon Kang}
\affil{Department of Computer Science and Statistics, Jeju National University}
\author[2]{Junmo Song\thanks{Corresponding author, e-mail: \href{mailto:jsong@knu.ac.kr}{jsong@knu.ac.kr}}}
\affil{Department of Statistics, Kyungpook National University}
\date{}

\begin{document}
\maketitle

\begin{abstract}
Parameter change test has been an important issue in time series analysis.  The problem has also been actively explored in the field of
integer-valued time series, but the testing in the presence of outliers has not yet been extensively investigated.
This study considers the problem of testing for parameter change in Poisson autoregressive
models particularly when observations are contaminated by outliers.
To lessen the impact of outliers on testing procedure, we propose a test based on the
density power divergence, which is introduced by Basu et al. (Biometrika, 1998),
and derive its limiting null distribution. Monte Carlo simulation results demonstrate validity and
strong robustness of the proposed test.\\
\end{abstract}


\section{Introduction}
Recently, there has been a growing interest in time series of counts because such data are frequently encountered in various application fields, for instance, epidemiology (Zeger (1988) and Jung and Tremayne (2011)), finance (Diop and Kengne (2017)), insurance industry (Zhu and Joe (2006)), statistical quality control (Weiß (2009)) and so on. Accordingly, integer-valued models for count time series have been developed by several authors. Among them, integer-valued generalized autoregressive conditional heteroskedastic (INGARCH) model introduced by Ferland {\it et al.} (2006) has been popularly used in practical applications, mainly because it can successfully capture the overdispersion phenomenon that occurs frequently in count time series. For more details, see Weiß (2010). INGARCH process follows the Poisson distribution conditionally on the past with the mean process, say $\{\lambda_t\}$, which is a linear function of its past values and past observations. Fokianos {\it et al.} (2009) generalized it to Poisson autoregressive (AR) models by allowing for nonlinearity of $\lambda_t$.

In this study, we consider the problem of testing for parameter change in Poisson AR models, particularly when outliers are involved in data.
Change point problem has attracted considerable attention from researchers and practitioners   and vast amount of literature
have been devoted to this area. For historical background and general review, we refer the readers to Aue and Horv\'{a}th (2013) and
Horv\'{a}th and Rice (2014). The problem has also been actively explored in the field of integer-valued time series by many authors. For example, Kang and Lee (2009, 2014a) proposed CUSUM tests for parameter change in RCINAR models and Poisson AR models, respectively, and  Doukhan and Kengne (2015) developed test procedures in a general class of Poisson AR models. Kang and Song (2017) introduced score test in Poisson AR models and Hudecov\'{a} {et al.} (2017) proposed probability generating function based methods for INAR models and Poisson AR models. However, to the best of our knowledge, testing for parameter change in the presence of outliers has not been dealt with except for Kang and Song (2015).

The study by Kang and Song (2015) addressed that sizes of estimates-based CUSUM test in Kang and Lee (2014a) for Poisson AR models are severely distorted by outliers and proposed a robust test using minimum density power divergence estimator (MDPDE). This distortions in sizes are also likely to be observed in other testing situations because the existing tests for time series models of counts are usually constructed based on (quasi-)likelihood estimator and such likelihood based methods are sensitive to outliers. Hence, various robust methods for integer-valued time series models also need to be developed for reliable inference in the presence of outliers.

In this paper, we introduce a density power (DP) divergence version of the score test in Kang and Song (2017). Indeed, the score test for parameter change firstly introduced by  Horv\'{a}th and Parzen (1994) can be considered as a test induced from Kullback-Leibler divergence. Recently,  Song and Kang (2019) investigated the extension of the score test  to DP divergence version and addressed that the proposed test in their study enjoys the merits of score test and DP divergence based inferences. That is, like the score test, the DP version of score test produces stable sizes when true parameter lies near the boundary of parameter space and, at the same time, it has robust and efficient properties as do other inferences based on DP divergence, see, for example, Basu et al. (1998) and Basu et al. (2016). It is importantly noted that the robust test by Kang and Song (2015) suffers from size distortions in such boundary situations like other estimates-based test such as Kang and Lee (2014a). Also, from a computational point of view, the test uses many partial MDPD estimates. Here, it should also be noted that optimization of the objective function of MDPDE is much more computationally burdensome. The proposed test in the present study can remedy such defects of MDPD estimates-based CUSUM test, while maintaining the robust property.

There have been notable studies that deal with robust tests using some divergences. Lee and Na (2005) introduced MDPD estimates-based CUSUM test and Batsidis et al. (2013) considered  a change-point detection using $\phi$-divergence. Recently, Basu {\it et al.} (2016) proposed generalized Wald-type tests based on DP divergence and Knoblauch et al. (2018) proposed a Bayesian online change point detection algorithm using $\beta$-divergence. For statistical inference based on various divergences, we refer the reader to Pardo (2006).

This paper is organized as follows. In Section 2, we review the MDPDE for Poisson AR models and its asymptotic properties.
In Section 3, we propose a robust test based on  DP divergence and derive its asymptotic null distribution. In Section 4, we perform a simulation study to compare with the score test. Section 5 concludes the paper. All the proofs for the results in Section 3 are provided in the Appendix.


\section{MDPDE for Poisson AR models}

In this section, we briefly review the MDPDE for Poisson AR models.
The Poisson AR model is defined by
\begin{eqnarray}\label{ACP}
X_t|\mathcal{F}_{t-1}\sim {Poisson}(\lambda_t),~~\lambda_t = f_\theta (\lambda_{t-1},X_{t-1})~~{\rm for\ all}\ t\in \mathbb{Z},
\end{eqnarray}
where $f_\theta$ is a known positive function on $[0,\infty)\times \mathbb{N}_0, \mathbb{N}_0=\mathbb{N}\cup\{0\}$, depending on
unknown parameter $\theta\in \Theta\subset \mathbb{R}^d$, and $\mathcal{F}_{t-1}$ is the $\sigma$-field generated by
$\{X_{t-1},X_{t-2},\ldots\}$.

Denote the unknown true parameter by $\theta_0$. To estimate $\theta_0$ in the presence of outliers, Kang and Lee (2014b) introduced the MDPDE for Poisson AR models as a robust estimator. The estimator is shown to have strong robustness with little loss in asymptotic efficiency. Suppose that $X_1,\cdots,X_n$ are observed from $(\ref{ACP})$. Then, the MDPDE for (\ref{ACP}) is given by
\begin{eqnarray*}\label{MDPDE1}
\hat\theta_{\alpha,n}=\argmin_{\theta\in\Theta}\tilde
H_{\alpha,n}(\theta)=\argmin_{\theta\in\Theta}\sum_{t=1}^{n}\tilde
l_{\alpha,t}(\theta),
\end{eqnarray*}
where
\begin{eqnarray*}
\tilde l_{\alpha,t}(\theta) := \left\{ \begin{array}{ll}
   \displaystyle  \sum_{y=0}^{\infty}\left(\frac{e^{-\tilde\lambda_t}\tilde\lambda_t^y}{y!}\right)^{1+\alpha}
-\left(1+\frac{1}{\alpha}\right)\left(\frac{e^{-\tilde\lambda_t}\tilde\lambda_t^{X_t}}{X_t!}\right)^{\alpha} & \mbox{, $\alpha > 0$\,,}\vspace{0.15cm}\\
   \displaystyle  \tilde\lambda_t - X_t {\rm log}\tilde\lambda_t + {\rm log}(X_t!)      & \mbox{, $\alpha = 0$,}
   \end{array}
 \right.
\end{eqnarray*}
and $\tilde\lambda_t$ are defined recursively by
\begin{eqnarray*}
\tilde \lambda_t =f_\theta (\tilde\lambda_{t-1},X_{t-1}),~~t \geq 2,
\end{eqnarray*}
with  arbitrarily chosen $\tilde\lambda_1$. Note that the MDPDE with $\A$=0 is exactly the same as the MLE.

In what follows, we denote by $l_{\alpha,t}(\theta)$ the counterpart of $\tilde l_{\alpha,t}(\theta)$
substituting $\tilde\lambda_t$ with $\lambda_t$. We use the notations $\tilde\lambda_t(\theta)$ and $\lambda_t(\theta)$ to
represent $\tilde\lambda_t$ and $\lambda_t$, respectively. Further, $\partial_\theta$ and $\partial_\theta^2$ are used to denote $\partial/\partial\theta$ and $\partial^2/\partial\theta\partial{\theta}^T$, respectively. The symbol $||\cdot||$ denotes the $l_2$ norm for matrices and vectors, and
$E(\cdot)$ is taken under $\theta_0$.

To derive asymptotic results for the MDPDE, the following assumptions are required :
\begin{enumerate}
\item[\bf A1.]  For all $\theta\in\Theta$,
$|f_\theta (\lambda,x)-f_\theta (\lambda',x')|
\leq \kappa_1 |\lambda-\lambda'|+ \kappa_2 |x-x'|~~{\rm~for~ all~} \lambda,\lambda'\geq 0 {\rm~and~} x,x'\in \mathbb{N}_0$, \\
where $\kappa_1$ and $\kappa_2$ are nonnegative real numbers with
$\kappa:=\kappa_1+\kappa_2<1$.
\item[\bf A2.] $\theta_0\in \Theta$ and $\Theta$ is compact. Also, for some $\delta_L>0$, the function $f$ satisfies
\begin{eqnarray*}
f_\theta(\lambda,x)\geq \delta_L
\end{eqnarray*}
for all $\theta\in \Theta$, $\lambda\geq 0$ and $x \in \mathbb{N}_0$.
\item[\bf A3.] $E(\sup_{\theta
\in\Theta}\lambda_1(\theta))<\infty$ and
     $E(\sup_{\theta \in\Theta}\tilde\lambda_1(\theta))<\infty$.
\item[\bf A4.]$\lambda_t(\theta)=\lambda_t(\theta_0)~a.s.$ implies $\theta=\theta_0$.
\item[\bf A5.] $\theta_0$ is an interior point of $\Theta$.
\item[\bf A6.] $\lambda_t(\theta)$ is twice continuously
differentiable with respect to $\theta$ and satisfies
     $$E\left(\sup_{\theta \in\Theta}\left|\left|\partial_\theta
     \lambda_t(\theta)\right|\right|\right)^4<\infty~~{\rm
     and}~~
     E\left(\sup_{\theta \in\Theta}\left|\left|\partial_\theta^2 \lambda_t(\theta)\right|\right|\right)^2<\infty.$$
\item[\bf A7.] There exists an integrable random variable $V$ and a real number $\rho$ with $0<\rho<1$, such that, a.s.,\\
     $$\sup_{\theta\in\Theta}\left|\left|\partial_\theta \lambda_t(\theta)-\partial_\theta \tilde\lambda_t(\theta)
     \right|\right|\leq V\rho^{t}~~{\rm and}~~
     \sup_{\theta \in\Theta}\left|\left|\partial_\theta^2 \lambda_t(\theta)-\partial_\theta^2 \tilde\lambda_t(\theta)\right|\right|\leq V\rho^{t}.$$
\item[\bf A8.] $\nu^T\partial_\theta \lambda_t(\theta_0)=0 ~a.s.$ implies  $\nu= 0$.
\end{enumerate}

Under {\bf A1}, there is a strictly stationary and ergodic solution for $(\ref{ACP})$ and
any order moments of $X_t$ and $\lambda_t$ are finite (cf. Neumann (2011) and Doukhan {\it et al.} (2012)).
The following asymptotic result is established by Kang and Lee (2014b).


\begin{theorem}\label{thm1} Under {\bf A1}-{\bf A4}, for each $\alpha \geq 0$, $\hat\theta_{\alpha,n}$ converges almost surely to $\theta_0$. If, in addition, {\bf A5}-{\bf A8} hold, then
\begin{eqnarray*}
\sqrt{n}(\hat\theta_{\alpha,n}-\theta_0)\stackrel{d}{\longrightarrow}N\left(0,J_{\alpha}^{-1}K_{\alpha}
J_{\alpha}^{-1}\right)~~as~ n\rightarrow\infty,
\end{eqnarray*}
where $K_\alpha:=(1+\alpha)^{-2}E\left(\partial_\theta l_{\alpha,t}(\theta_0) \partial_{\theta^T} l_{\alpha,t}(\theta_0)\right)$ and $J_\alpha:=-(1+\alpha)^{-1} E\left(\partial_\theta^2 l_{\alpha,t}(\theta_0)\right)$.
\end{theorem}

\begin{remark}
Model $(\ref{ACP})$ with a linear specification $f(\lambda, x)=w+a\lambda+b x$ is referred to as INGARCH(1,1) models. The INGARCH model satisfies {\bf A1} when $a+b<1$. This model is particularly attractive for overdispersed count data. For more details, see Ferland {\it et al.} (2006).
\end{remark}



\section{DP divergence based test for parameter change in Poisson AR models}
In this section, we consider the problem of testing the following hypotheses in the presence of outliers:
\begin{eqnarray*}
&& H_0: \theta\textrm{ does not change over~} X_1,\ldots,X_n~ vs.\\
&& H_1:\textrm{not~} H_0.
\end{eqnarray*}
For this, we construct a test statistics using the estimating function of the MDPDE, i.e., $\pa_{\theta}\tilde{H}_{\A,n}(\T)$.
By applying Taylor's theorem to $\pa_{\theta}\tilde{H}_{\A,n}(\T)$, we have that  for each $ s \in [0,1]$,
\begin{eqnarray}\label{Taylor1}
\frac{1}{\sqrt{n}}\pa_{\theta}\tilde{H}_{\A,[ns]}(\HT)=\frac{1}{\sqrt{n}}\pa_{\theta}\tilde{H}_{\A,[ns]}(\T_0)+\frac{1}{n}\paa \tilde{H}_{\A,[ns]}(\T^*_{\A,n,s}) \sqrt{n}(\HT-\theta_0),
\end{eqnarray}
where $\T^*_{\A,n,s}$ is an intermediate point between $\theta_0$ and $\HT$. Noting the fact that $\pa_{\T}\tilde{H}_{\A,n}(\HT)=0$, we also have that for $s=1$,
\begin{eqnarray*}
0=\frac{1}{\sqrt{n}}\pa_{\T}\tilde{H}_{\A,n}(\HT)=\frac{1}{\sqrt{n}}\,\pa_{\T}\tilde{H}_{\A,n}(\T_0)+\frac{1}{n}\paa \tilde{H}_{\A,n}(\T^*_{\A,n,1}) \sqrt{n}(\HT-\T_0),
\end{eqnarray*}
and thus it can be written that
\begin{eqnarray}\label{Taylor2}
\sqrt{n}(\HT-\T_0)= \frac{1}{1+\A}J_\A^{-1}\frac{1}{\sqrt{n}}\,\pa_{\T}\tilde{H}_{\A,n}(\T_0)+\frac{1}{1+\A}J_\A^{-1}(B_{\A,n}+(1+\A)J_\A)\sqrt{n}(\HT-\T_0),
\end{eqnarray}
where $B_{\A,n}=\paa\tilde{H}_{\A,n}(\T^*_{\A,n,1})/n$. Hence, by substituting  $(\ref{Taylor2})$ into  $(\ref{Taylor1})$, we can express that
\begin{eqnarray*}
\frac{1}{\sqrt{n}}\pa_{\theta}\tilde{H}_{\A,[ns]}(\HT)=I_n+II_n+III_n
\end{eqnarray*}
where
\begin{eqnarray*}
I_n&:=&\frac{1}{\sqrt{n}}\pa_{\theta}\tilde{H}_{\A,[ns]}(\T_0)
-\frac{[ns]}{n}\frac{1}{\sqrt{n}}\,\pa_{\theta}\tilde{H}_{\A,n}(\theta_0)\\
II_n&:=&\frac{1}{n}\paa\tilde{H}_{\A,[ns]}(\T^*_{\A,n,s})
\frac{1}{1+\A}J_\A^{-1}\frac{1}{\sqrt{n}}\,\pa_{\T}\tilde{H}_{\A,n}(\T_0)
+\frac{[ns]}{n}\frac{1}{\sqrt{n}}\,\pa_{\theta}\tilde{H}_{\A,n}(\theta_0)\\
III_n&:=&\frac{1}{n}\paa \tilde{H}_{\A,[ns]}(\T^*_{\A,n,s})
\frac{1}{1+\A}J_\A^{-1}(B_{\A,n}+(1+\A)J_\A)\sqrt{n}(\HT-\T_0).
\end{eqnarray*}
We first note that by Lemma \ref{Lm1} in the Appendix,
\[\frac{1}{1+\A}K_\A^{-1/2} I_n \stackrel{w}{\longrightarrow}
 B^o_d(s)\quad \rm{in}\ \ \mathbb{D}\,\big( [0,1],\, \mathbb{R}^d\big)\]
where $B^o_d$ is a $d$-dimensional standard Brownian bridge. Furthermore, due to  Lemmas \ref{Lm4} and \ref{Lm5} below, one can see that $II_n$ and $III_n$ are asymptotically negligible, respectively. Hence, combining the above arguments, we obtain the following main result.


\begin{theorem}\label{Thm_Score}
Suppose that the assumptions {\bf A1}- {\bf A8} hold. Then, under $H_0$, we
have
\begin{eqnarray*}
\frac{1}{1+\A}K_\A^{-1/2}\frac{1}{\sqrt{n}}\pa_{\theta}\tilde{H}_{\A,[ns]}(\HT)\stackrel{w}{\longrightarrow}\,B^o_d(s)\quad
\rm{in}\ \ \mathbb{D}\,\big( [0,1],\, \mathbb{R}^d\big)\,,
\end{eqnarray*}
thus
\begin{eqnarray*}
T_{n}^{\A}:=\frac{1}{(1+\A)^2}\max_{1\leq k \leq n}
\frac{1}{n}\pa_{\theta}\tilde{H}_{\A,k}(\HT)^T \hat{K}_\A^{-1}
\pa_{\theta}\tilde{H}_{\A,k}(\HT)\stackrel{d}{\longrightarrow}
\sup_{0\leq s \leq 1} \big\|B_d^o(s)\big\|_2^2,
\end{eqnarray*}
where
$\hat{K}_\A$ is a consistent estimator of $K_\A$.
We reject $H_0$ if $T_{n}^{\A}$ is large.
\end{theorem}

\begin{remark}
As a consistent estimator of $K_\A$, one can consider to use
\[\hat K_{\alpha}=\frac{1}{(1+\alpha)^2}\frac{1}{n}\sum_{t=1}^{n} \partial_\theta \tilde l_{\alpha,t}(\hat\theta_{\alpha,n})
\partial_{\theta^T} \tilde l_{\alpha,t}(\hat\theta_{\alpha,n}).\]
For the consistency of the estimator, see Kang and Song (2015).
\end{remark}

\begin{remark}\label{RM3}
Since the MDPDE with $\alpha$=0 is the MLE,
$T_n^{\A}$ with $\A$=0 becomes the score test in Kang and Song (2017) given by
\[T_n:=\max_{1\leq k \leq n}\frac{1}{n}\pa_{\theta}\tilde{L}_{k}(\hat{\theta}_n)^T \hat I_n^{-1}
\pa_{\theta}\tilde{L}_{k}(\hat{\theta}_n)\stackrel{d}{\longrightarrow}
\sup_{0\leq s \leq 1} \big\|B_d^o(s)\big\|_2^2, \]
where $\pa_{\theta}\tilde{L}_{n}(\T)$ is the score function and $\hat I_n$ is a consistent estimator of the information matrix.
\end{remark}


\section{Simulation study}
In this section, we evaluate the performance of $T_{n}^{\A}$ with $\alpha>0$ and compare it with the score test $T_{n}$ in Remark \ref{RM3}.
For this task, we consider the following INGARCH(1,1) models :
\begin{eqnarray}\label{INGARCH}
X_t|\mathcal{F}_{t-1}\sim Poisson(\lambda_t),~~\lambda_t= w +a
\lambda_{t-1}+b X_{t-1},
\end{eqnarray}
where $\lambda_1$ is assumed to be 0.  The sample sizes under consideration are $n$=300 and 500.
The method of moment estimates are
used as initial estimates for optimization procedure. For each simulation, the first 1,000 initial observations are discarded to avoid initialization effects.
The empirical sizes and powers are calculated as the proportion of the number of rejections of the null hypothesis based on 1,000 repetitions. The critical values corresponding to the nominal level 5\% and 10\% are 3.027 and 2.604, respectively, which are obtained through Monte Carlo simulations.

\begin{table}[t]
  \centering
    {\footnotesize
  \tabcolsep=3.2pt
  \caption{Empirical sizes of $T_n$ and $T_n^{\A}$ with no outliers.}\vspace{0cm}
\begin{tabular}{ccccccccccccccccccccc}
\toprule
      &       &       &       & \multicolumn{2}{c}{\multirow{2}[4]{*}{$T_n$}} &       & \multicolumn{14}{c}{$T_n^\alpha$ } \\
\cmidrule{8-21}      &       &       &       & \multicolumn{2}{c}{} &       & \multicolumn{2}{c}{$\alpha=0.1$} &       & \multicolumn{2}{c}{$\alpha=0.2$} &       & \multicolumn{2}{c}{$\alpha=0.3$} &       & \multicolumn{2}{c}{$\alpha=0.5$} &       & \multicolumn{2}{c}{$\alpha=1$} \\
\cmidrule{5-6}\cmidrule{8-9}\cmidrule{11-12}\cmidrule{14-15}\cmidrule{17-18}\cmidrule{20-21}$\theta$ &       & $n$   &       & 5\%   & 10\%  &       & 5\%   & 10\%  &       & 5\%   & 10\%  &       & 5\%   & 10\%  &       & 5\%   & 10\%  &       & 5\%   & 10\% \\
\midrule
(2,0.1,0.2) &       & 300   &       & 0.062 & 0.106 &       & 0.066 & 0.110 &       & 0.065 & 0.110 &       & 0.062 & 0.108 &       & 0.065 & 0.110 &       & 0.070 & 0.110 \\
      &       & 500   &       & 0.060 & 0.114 &       & 0.064 & 0.115 &       & 0.060 & 0.111 &       & 0.058 & 0.110 &       & 0.060 & 0.113 &       & 0.065 & 0.114 \\
\cmidrule{3-21}(2,0.1,0.4) &       & 300   &       & 0.056 & 0.096 &       & 0.051 & 0.104 &       & 0.052 & 0.104 &       & 0.053 & 0.104 &       & 0.051 & 0.103 &       & 0.050 & 0.100 \\
      &       & 500   &       & 0.056 & 0.091 &       & 0.050 & 0.096 &       & 0.056 & 0.096 &       & 0.056 & 0.090 &       & 0.056 & 0.099 &       & 0.060 & 0.100 \\
\cmidrule{3-21}(2,0.1,0.7) &       & 300   &       & 0.030 & 0.072 &       & 0.048 & 0.104 &       & 0.050 & 0.106 &       & 0.051 & 0.104 &       & 0.052 & 0.100 &       & 0.054 & 0.103 \\
      &       & 500   &       & 0.038 & 0.078 &       & 0.054 & 0.108 &       & 0.058 & 0.106 &       & 0.058 & 0.106 &       & 0.054 & 0.108 &       & 0.054 & 0.108 \\
\bottomrule
\end{tabular}%
  \label{tab1}\vspace{0.5cm}

   \tabcolsep=4pt
    \caption{Empirical powers of $T_n$ and $T_n^{\A}$ at nominal level 5\%  with no outliers.}\vspace{0cm}
    \tabcolsep=5pt
    \begin{tabular}{lccccccccc}
    \toprule
          &       &       & \multirow{2}[4]{*}{$T_n$} &       & \multicolumn{5}{c}{$T_n^\alpha$ } \\
\cmidrule{6-10}   $\theta\ \rightarrow\ \theta'$ & $n$     &       &       &       & 0.1 & 0.2 & 0.3 & 0.5 & 1.0 \\
    \midrule
    (2,0.1,0.2)$\rightarrow$(2.5,0.1,0.2) & 300   &       & 0.324 &       & 0.306 & 0.291 & 0.272 & 0.241 & 0.186 \\
          & 500   &       & 0.644 &       & 0.636 & 0.612 & 0.585 & 0.534 & 0.392 \\
\cmidrule{2-10}    (2,0.1,0.2)$\rightarrow$(2,0.3,0.2) & 300   &       & 0.502 &       & 0.472 & 0.454 & 0.430 & 0.378 & 0.272 \\
          & 500   &       & 0.835 &       & 0.822 & 0.803 & 0.772 & 0.706 & 0.535 \\
\cmidrule{2-10}    (2,0.1,0.2)$\rightarrow$(2,0.1,0.4) & 300   &       & 0.615 &       & 0.586 & 0.571 & 0.548 & 0.492 & 0.380 \\
          & 500   &       & 0.930 &       & 0.922 & 0.911 & 0.894 & 0.859 & 0.712 \\
    \bottomrule
    \end{tabular}%
  \label{tab2}}

\end{table}%

We first address the case in which the data are not contaminated by outliers. To calculate empirical sizes, observations are generated from the model $(\ref{INGARCH})$ with $\theta=(w,a,b)$=(2,0.1,0.2), (2,0.1,0.4), and (2,0.1,0.7). For the empirical powers, we consider the alternatives that $\theta$ changes from $(2,0.1,0.2)$ to $\theta'$=(2.5,0.1,0.2), (2,0.3,0.2) and (2,0.1,0.4) at the middle time $t=[n/2]$. Tables \ref{tab1} and \ref{tab2} provide the results for uncontaminated cases. From Table \ref{tab1}, we can see that both $T_n$ and $T_n^{\A}$ produce appropriate empirical
sizes. Even for the case that $a+b$ is close to unity, no size distortion is observed. Here, we note the MDPD and ML estimates-based CUSUM tests yielded distorted sizes in the same parameter setting, see Kang and Song (2015) for more details. It can also be seen from Table \ref{tab2} that
$T_n$ and $T_n^\A$ with $\A$ close to 0 produce reasonably good powers. The power of $T_n^\A$, however, shows a  tendency  to decrease with an increase in $\A$. As expected, $T_n$ shows best performance and $T_n^{\A}$ with $\A$ close to 0 performs similarly to $T_n$.

\begin{table}[]
\vspace{-0.9cm}
 \centering
   {\footnotesize
  \tabcolsep=3pt
  \renewcommand{\arraystretch}{0.95}
  \caption{Empirical sizes and $d_\A$ of $T_n$ and $T_n^{\A}$ at nominal level 5\% when $p$=0.01 and $\gamma$=10.}\vspace{0cm}
\begin{tabular}{ccccccccccccccccccccc}\toprule
      &       & $AO$  &       &       &       & \multicolumn{5}{c}{$T_n^\alpha$ }     &       & $IO$  &       &       &       & \multicolumn{5}{c}{$T_n^\alpha$ } \\
\cmidrule{7-11}\cmidrule{17-21} $\theta$ &       & $n$   &       & $T_n$ &       & 0.1   & 0.2   & 0.3   & 0.5   & 1.0   &       & $n$   &       & $T_n$ &       & 0.1   & 0.2   & 0.3   & 0.5   & 1.0 \\
\midrule
(2,0.1,0.2) &       & 300   &       & 0.152 &       & 0.063 & 0.068 & 0.069 & 0.065 & 0.066 &       & 300   &       & 0.158 &       & 0.062 & 0.060 & 0.058 & 0.062 & 0.058 \\
      &       &       &       & [2.45] &       & [0.95] & [1.05] & [1.11] & [1.00] & [0.94] &       &       &       & [2.55] &       & [0.94] & [0.92] & [0.94] & [0.95] & [0.83] \\
\cmidrule{3-11}\cmidrule{13-21}      &       & 500   &       & 0.172 &       & 0.068 & 0.071 & 0.070 & 0.068 & 0.065 &       & 500   &       & 0.196 &       & 0.076 & 0.071 & 0.069 & 0.068 & 0.068 \\
      &       &       &       & [2.87] &       & [1.06] & [1.18] & [1.21] & [1.13] & [1.00] &       &       &       & [3.27] &       & [1.19] & [1.18] & [1.19] & [1.13] & [1.05] \\
\midrule
(2,0.1,0.4) &       & 300   &       & 0.122 &       & 0.044 & 0.048 & 0.052 & 0.048 & 0.053 &       & 300   &       & 0.127 &       & 0.054 & 0.054 & 0.053 & 0.051 & 0.058 \\
      &       &       &       & [2.18] &       & [0.86] & [0.92] & [0.98] & [0.94] & [1.06] &       &       &       & [2.27] &       & [1.06] & [1.04] & [1.00] & [1.00] & [1.16] \\
\cmidrule{3-11}\cmidrule{13-21}      &       & 500   &       & 0.147 &       & 0.050 & 0.055 & 0.064 & 0.061 & 0.066 &       & 500   &       & 0.137 &       & 0.052 & 0.049 & 0.053 & 0.056 & 0.054 \\
      &       &       &       & [2.63] &       & [1.00] & [0.98] & [1.14] & [1.09] & [1.10] &       &       &       & [2.45] &       & [1.04] & [0.88] & [0.95] & [1.00] & [0.90] \\
\midrule
(2,0.1,0.7) &       & 300   &       & 0.088 &       & 0.049 & 0.052 & 0.054 & 0.055 & 0.060 &       & 300   &       & 0.072 &       & 0.047 & 0.051 & 0.052 & 0.052 & 0.052 \\
      &       &       &       & [2.93] &       & [1.02] & [1.04] & [1.06] & [1.06] & [1.11] &       &       &       & [2.40] &       & [0.98] & [1.02] & [1.02] & [1.00] & [0.96] \\
\cmidrule{3-11}\cmidrule{13-21}      &       & 500   &       & 0.086 &       & 0.056 & 0.055 & 0.058 & 0.058 & 0.060 &       & 500   &       & 0.078 &       & 0.054 & 0.059 & 0.057 & 0.058 & 0.056 \\
      &       &       &       & [2.26] &       & [1.04] & [0.95] & [1.00] & [1.07] & [1.11] &       &       &       & [2.05] &       & [1.00] & [1.02] & [0.98] & [1.07] & [1.04] \\
\bottomrule
 \multicolumn{21}{l}{The figures in the brackets represent $d_\A$.}
\end{tabular}
  \label{tab3}\vspace{0.25cm}

   \caption{Empirical sizes and $d_\A$ of $T_n$ and $T_n^{\A}$ at nominal level 5\%  when  $p$=0.01 and $\gamma$=20.}\vspace{0cm}
  \begin{tabular}{ccccccccccccccccccccc}
  \toprule
      &       & $AO$  &       &       &       & \multicolumn{5}{c}{$T_n^\alpha$ }     &       & $IO$  &       &       &       & \multicolumn{5}{c}{$T_n^\alpha$ } \\
\cmidrule{7-11}\cmidrule{17-21} $\theta$ &       & $n$   &       & $T_n$ &       & 0.1   & 0.2   & 0.3   & 0.5   & 1.0   &       & $n$   &       & $T_n$ &       & 0.1   & 0.2   & 0.3   & 0.5   & 1.0 \\
\midrule
(2,0.1,0.2) &       & 300   &       & 0.398 &       & 0.063 & 0.060 & 0.063 & 0.066 & 0.062 &       & 300   &       & 0.414 &       & 0.060 & 0.062 & 0.059 & 0.060 & 0.069 \\
      &       &       &       & [6.42] &       & [0.95] & [0.92] & [1.02] & [1.02] & [0.89] &       &       &       & [6.68] &       & [0.91] & [0.95] & [0.95] & [0.92] & [0.99] \\
\cmidrule{3-11}\cmidrule{13-21}      &       & 500   &       & 0.468 &       & 0.074 & 0.072 & 0.067 & 0.068 & 0.064 &       & 500   &       & 0.464 &       & 0.058 & 0.056 & 0.056 & 0.052 & 0.060 \\
      &       &       &       & [7.80] &       & [1.16] & [1.20] & [1.16] & [1.13] & [0.98] &       &       &       & [7.73] &       & [0.91] & [0.93] & [0.97] & [0.87] & [0.92] \\
\midrule
(2,0.1,0.4) &       & 300   &       & 0.407 &       & 0.068 & 0.070 & 0.069 & 0.068 & 0.065 &       & 300   &       & 0.353 &       & 0.048 & 0.055 & 0.056 & 0.059 & 0.065 \\
      &       &       &       & [7.27] &       & [1.33] & [1.35] & [1.30] & [1.33] & [1.03] &       &       &       & [6.30] &       & [0.94] & [1.06] & [1.06] & [1.16] & [1.30] \\
\cmidrule{3-11}\cmidrule{13-21}      &       & 500   &       & 0.444 &       & 0.074 & 0.071 & 0.068 & 0.062 & 0.064 &       & 500   &       & 0.356 &       & 0.050 & 0.056 & 0.056 & 0.050 & 0.060 \\
      &       &       &       & [7.93] &       & [1.48] & [1.27] & [1.21] & [1.11] & [1.07] &       &       &       & [6.36] &       & [1.00] & [1.00] & [1.00] & [0.89] & [1.00] \\
\midrule
(2,0.1,0.7) &       & 300   &       & 0.274 &       & 0.054 & 0.055 & 0.056 & 0.055 & 0.062 &       & 300   &       & 0.184 &       & 0.044 & 0.046 & 0.048 & 0.053 & 0.054 \\
      &       &       &       & [9.13] &       & [1.13] & [1.10] & [1.10] & [1.06] & [1.15] &       &       &       & [6.13] &       & [0.92] & [0.92] & [0.94] & [1.02] & [1.00] \\
\cmidrule{3-11}\cmidrule{13-21}      &       & 500   &       & 0.315 &       & 0.048 & 0.054 & 0.058 & 0.058 & 0.056 &       & 500   &       & 0.179 &       & 0.047 & 0.046 & 0.050 & 0.054 & 0.056 \\
      &       &       &       & [8.29] &       & [0.89] & [0.93] & [1.00] & [1.07] & [1.04] &       &       &       & [4.71] &       & [0.87] & [0.79] & [0.86] & [1.00] & [1.04] \\
\bottomrule
\end{tabular}
 \label{tab4}\vspace{0.25cm}
\caption{Empirical sizes and $d_\A$ of $T_n$ and $T_n^{\A}$ at nominal level 5\% when $p$=0.03 and $\gamma$=10.}\vspace{0cm}
\begin{tabular}{ccccccccccccccccccccc}\toprule
      &       & $AO$  &       &       &       & \multicolumn{5}{c}{$T_n^\alpha$ }     &       & $IO$  &       &       &       & \multicolumn{5}{c}{$T_n^\alpha$ } \\
\cmidrule{7-11}\cmidrule{17-21} $\theta$ &       & $n$   &       & $T_n$ &       & 0.1   & 0.2   & 0.3   & 0.5   & 1.0   &       & $n$   &       & $T_n$ &       & 0.1   & 0.2   & 0.3   & 0.5   & 1.0 \\
\midrule
(2,0.1,0.2) &       & 300   &       & 0.370 &       & 0.066 & 0.070 & 0.071 & 0.068 & 0.064 &       & 300   &       & 0.372 &       & 0.064 & 0.059 & 0.064 & 0.064 & 0.068 \\
      &       &       &       & [5.97] &       & [1.00] & [1.08] & [1.15] & [1.05] & [0.91] &       &       &       & [6.00] &       & [0.97] & [0.91] & [1.03] & [0.98] & [0.97] \\
\cmidrule{3-11}\cmidrule{13-21}      &       & 500   &       & 0.356 &       & 0.071 & 0.075 & 0.070 & 0.067 & 0.070 &       & 500   &       & 0.402 &       & 0.067 & 0.062 & 0.062 & 0.064 & 0.068 \\
      &       &       &       & [5.93] &       & [1.11] & [1.25] & [1.21] & [1.12] & [1.08] &       &       &       & [6.70] &       & [1.05] & [1.03] & [1.07] & [1.07] & [1.05] \\
\midrule
(2,0.1,0.4) &       & 300   &       & 0.335 &       & 0.070 & 0.066 & 0.068 & 0.066 & 0.062 &       & 300   &       & 0.300 &       & 0.054 & 0.050 & 0.052 & 0.058 & 0.060 \\
      &       &       &       & [5.98] &       & [1.37] & [1.27] & [1.28] & [1.29] & [1.24] &       &       &       & [5.36] &       & [1.06] & [0.96] & [0.98] & [1.14] & [1.20] \\
\cmidrule{3-11}\cmidrule{13-21}      &       & 500   &       & 0.304 &       & 0.046 & 0.054 & 0.054 & 0.057 & 0.057 &       & 500   &       & 0.310 &       & 0.048 & 0.046 & 0.046 & 0.047 & 0.046 \\
      &       &       &       & [5.43] &       & [0.92] & [0.96] & [0.96] & [1.02] & [0.95] &       &       &       & [5.54] &       & [0.96] & [0.82] & [0.82] & [0.84] & [0.77] \\
\midrule
(2,0.1,0.7) &       & 300   &       & 0.203 &       & 0.062 & 0.065 & 0.070 & 0.064 & 0.058 &       & 300   &       & 0.146 &       & 0.050 & 0.046 & 0.048 & 0.052 & 0.054 \\
      &       &       &       & [6.77] &       & [1.29] & [1.30] & [1.37] & [1.23] & [1.07] &       &       &       & [4.87] &       & [1.04] & [0.92] & [0.94] & [1.00] & [1.00] \\
\cmidrule{3-11}\cmidrule{13-21}      &       & 500   &       & 0.199 &       & 0.045 & 0.050 & 0.056 & 0.058 & 0.054 &       & 500   &       & 0.160 &       & 0.043 & 0.045 & 0.044 & 0.044 & 0.050 \\
      &       &       &       & [5.24] &       & [0.83] & [0.86] & [0.97] & [1.07] & [1.00] &       &       &       & [4.21] &       & [0.80] & [0.78] & [0.76] & [0.81] & [0.93] \\
\bottomrule
\end{tabular}
 \label{tab5}

}
\end{table}%

\begin{table}[]
\vspace{-0.9cm}
  \centering
   {\footnotesize
  \tabcolsep=3pt
  \renewcommand{\arraystretch}{0.95}
  \caption{Empirical powers and $d_\A$ of $T_n$ and $T_n^{\A}$ at nominal level 5\% when $p$=0.01 and $\gamma$=10.}\vspace{0cm}
\begin{tabular}{lcccccccccccccccccccc}
\toprule
      &       & $AO$  &       &       &       & \multicolumn{5}{c}{$T_n^\alpha$ }     &       & $IO$  &       &       &       & \multicolumn{5}{c}{$T_n^\alpha$ } \\
\cmidrule{7-11}\cmidrule{17-21} $\theta\ \rightarrow\ \theta'$ &       & $n$   &       & $T_n$ &       & 0.1   & 0.2   & 0.3   & 0.5   & 1.0   &       & $n$   &       & $T_n$ &       & 0.1   & 0.2   & 0.3   & 0.5   & 1.0 \\
\midrule
(2,0.1,0.2) &       & 300   &       & 0.436 &       & 0.302 & 0.313 & 0.304 & 0.278 & 0.216 &       & 300   &       & 0.440 &       & 0.300 & 0.293 & 0.282 & 0.250 & 0.192 \\
$\rightarrow$(2.5,0.1,0.2) &       &       &       & [1.35] &       & [0.99] & [1.08] & [1.12] & [1.15] & [1.16] &       &       &       & [1.36] &       & [0.98] & [1.01] & [1.04] & [1.04] & [1.03] \\
\cmidrule{3-11}\cmidrule{13-21}      &       & 500   &       & 0.718 &       & 0.624 & 0.635 & 0.610 & 0.556 & 0.418 &       & 500   &       & 0.708 &       & 0.630 & 0.638 & 0.610 & 0.554 & 0.418 \\
      &       &       &       & [1.11] &       & [0.98] & [1.04] & [1.04] & [1.04] & [1.07] &       &       &       & [1.10] &       & [0.99] & [1.04] & [1.04] & [1.04] & [1.07] \\
\midrule
(2,0.1,0.2) &       & 300   &       & 0.596 &       & 0.491 & 0.500 & 0.481 & 0.436 & 0.322 &       & 300   &       & 0.672 &       & 0.556 & 0.548 & 0.522 & 0.456 & 0.329 \\
$\rightarrow$(2,0.3,0.2) &       &       &       & [1.19] &       & [1.04] & [1.10] & [1.12] & [1.15] & [1.18] &       &       &       & [1.34] &       & [1.18] & [1.21] & [1.21] & [1.21] & [1.21] \\
\cmidrule{3-11}\cmidrule{13-21}      &       & 500   &       & 0.860 &       & 0.811 & 0.814 & 0.802 & 0.756 & 0.610 &       & 500   &       & 0.928 &       & 0.910 & 0.904 & 0.884 & 0.835 & 0.671 \\
      &       &       &       & [1.03] &       & [0.99] & [1.01] & [1.04] & [1.07] & [1.14] &       &       &       & [1.11] &       & [1.11] & [1.13] & [1.15] & [1.18] & [1.25] \\
\midrule
(2,0.1,0.2) &       & 300   &       & 0.700 &       & 0.547 & 0.554 & 0.549 & 0.515 & 0.400 &       & 300   &       & 0.702 &       & 0.590 & 0.589 & 0.570 & 0.524 & 0.402 \\
$\rightarrow$(2,0.1,0.4) &       &       &       & [1.14] &       & [0.93] & [0.97] & [1.00] & [1.05] & [1.05] &       &       &       & [1.14] &       & [1.01] & [1.03] & [1.04] & [1.07] & [1.06] \\
\cmidrule{3-11}\cmidrule{13-21}      &       & 500   &       & 0.954 &       & 0.900 & 0.902 & 0.890 & 0.856 & 0.720 &       & 500   &       & 0.938 &       & 0.926 & 0.926 & 0.922 & 0.886 & 0.767 \\
      &       &       &       & [1.03] &       & [0.98] & [0.99] & [1.00] & [1.00] & [1.01] &       &       &       & [1.01] &       & [1.00] & [1.02] & [1.03] & [1.03] & [1.08] \\
\bottomrule
 \multicolumn{21}{l}{The figures in the brackets represent $d_\A$.}
\end{tabular}%
  \label{tab6}\vspace{0.25cm}

   \caption{Empirical powers and $d_\A$ of $T_n$ and $T_n^{\A}$ at nominal level 5\%  when  $p$=0.01 and $\gamma$=20.}\vspace{0cm}
 \begin{tabular}{lcccccccccccccccccccc}
\toprule
      &       & $AO$  &       &       &       & \multicolumn{5}{c}{$T_n^\alpha$ }     &       & $IO$  &       &       &       & \multicolumn{5}{c}{$T_n^\alpha$ } \\
\cmidrule{7-11}\cmidrule{17-21} $\theta\ \rightarrow\ \theta'$ &       & $n$   &       & $T_n$ &       & 0.1   & 0.2   & 0.3   & 0.5   & 1.0   &       & $n$   &       & $T_n$ &       & 0.1   & 0.2   & 0.3   & 0.5   & 1.0 \\
\midrule
(2,0.1,0.2) &       & 300   &       & 0.650 &       & 0.420 & 0.421 & 0.398 & 0.346 & 0.251 &       & 300   &       & 0.609 &       & 0.348 & 0.348 & 0.327 & 0.294 & 0.216 \\
$\rightarrow$(2.5,0.1,0.2) &       &       &       & [2.01] &       & [1.37] & [1.45] & [1.46] & [1.44] & [1.35] &       &       &       & [1.88] &       & [1.14] & [1.20] & [1.20] & [1.22] & [1.16] \\
\cmidrule{3-11}\cmidrule{13-21}      &       & 500   &       & 0.842 &       & 0.730 & 0.722 & 0.702 & 0.636 & 0.472 &       & 500   &       & 0.800 &       & 0.642 & 0.633 & 0.617 & 0.566 & 0.441 \\
      &       &       &       & [1.31] &       & [1.15] & [1.18] & [1.20] & [1.19] & [1.20] &       &       &       & [1.24] &       & [1.01] & [1.03] & [1.05] & [1.06] & [1.13] \\
\midrule
(2,0.1,0.2) &       & 300   &       & 0.802 &       & 0.674 & 0.670 & 0.640 & 0.572 & 0.421 &       & 300   &       & 0.813 &       & 0.660 & 0.658 & 0.622 & 0.559 & 0.395 \\
$\rightarrow$(2,0.3,0.2) &       &       &       & [1.60] &       & [1.43] & [1.48] & [1.49] & [1.51] & [1.55] &       &       &       & [1.62] &       & [1.40] & [1.45] & [1.45] & [1.48] & [1.45] \\
\cmidrule{3-11}\cmidrule{13-21}      &       & 500   &       & 0.950 &       & 0.900 & 0.898 & 0.882 & 0.850 & 0.721 &       & 500   &       & 0.966 &       & 0.961 & 0.952 & 0.936 & 0.901 & 0.760 \\
      &       &       &       & [1.14] &       & [1.09] & [1.12] & [1.14] & [1.20] & [1.35] &       &       &       & [1.16] &       & [1.17] & [1.19] & [1.21] & [1.28] & [1.42] \\
\midrule
(2,0.1,0.2) &       & 300   &       & 0.860 &       & 0.668 & 0.660 & 0.644 & 0.592 & 0.438 &       & 300   &       & 0.818 &       & 0.674 & 0.668 & 0.650 & 0.599 & 0.465 \\
$\rightarrow$(2,0.1,0.4) &       &       &       & [1.40] &       & [1.14] & [1.16] & [1.18] & [1.20] & [1.15] &       &       &       & [1.33] &       & [1.15] & [1.17] & [1.19] & [1.22] & [1.22] \\
\cmidrule{3-11}\cmidrule{13-21}      &       & 500   &       & 0.972 &       & 0.915 & 0.914 & 0.892 & 0.856 & 0.743 &       & 500   &       & 0.951 &       & 0.945 & 0.946 & 0.940 & 0.916 & 0.798 \\
      &       &       &       & [1.05] &       & [0.99] & [1.00] & [1.00] & [1.00] & [1.04] &       &       &       & [1.02] &       & [1.02] & [1.04] & [1.05] & [1.07] & [1.12] \\
\bottomrule
\end{tabular}%
 \label{tab7}\vspace{0.25cm}
\caption{Empirical powers and $d_\A$ of $T_n$ and $T_n^{\A}$ at nominal level 5\% when $p$=0.03 and $\gamma$=10.}\vspace{0cm}
\begin{tabular}{lcccccccccccccccccccc}
\toprule
      &       & $AO$  &       &       &       & \multicolumn{5}{c}{$T_n^\alpha$ }     &       & $IO$  &       &       &       & \multicolumn{5}{c}{$T_n^\alpha$ } \\
\cmidrule{7-11}\cmidrule{17-21} $\theta\ \rightarrow\ \theta'$ &       & $n$   &       & $T_n$ &       & 0.1   & 0.2   & 0.3   & 0.5   & 1.0   &       & $n$   &       & $T_n$ &       & 0.1   & 0.2   & 0.3   & 0.5   & 1.0 \\
\midrule
(2,0.1,0.2) &       & 300   &       & 0.624 &       & 0.326 & 0.360 & 0.362 & 0.344 & 0.254 &       & 300   &       & 0.608 &       & 0.292 & 0.306 & 0.300 & 0.276 & 0.212 \\
$\rightarrow$(2.5,0.1,0.2) &       &       &       & [1.93] &       & [1.07] & [1.24] & [1.33] & [1.43] & [1.37] &       &       &       & [1.88] &       & [0.95] & [1.05] & [1.10] & [1.15] & [1.14] \\
\cmidrule{3-11}\cmidrule{13-21}      &       & 500   &       & 0.805 &       & 0.606 & 0.659 & 0.644 & 0.602 & 0.488 &       & 500   &       & 0.769 &       & 0.578 & 0.604 & 0.593 & 0.544 & 0.424 \\
      &       &       &       & [1.25] &       & [0.95] & [1.08] & [1.10] & [1.13] & [1.24] &       &       &       & [1.19] &       & [0.91] & [0.99] & [1.01] & [1.02] & [1.08] \\
\midrule
(2,0.1,0.2) &       & 300   &       & 0.787 &       & 0.530 & 0.589 & 0.584 & 0.542 & 0.401 &       & 300   &       & 0.819 &       & 0.600 & 0.626 & 0.616 & 0.554 & 0.396 \\
$\rightarrow$(2,0.3,0.2) &       &       &       & [1.57] &       & [1.12] & [1.30] & [1.36] & [1.43] & [1.47] &       &       &       & [1.63] &       & [1.27] & [1.38] & [1.43] & [1.47] & [1.46] \\
\cmidrule{3-11}\cmidrule{13-21}      &       & 500   &       & 0.956 &       & 0.840 & 0.864 & 0.860 & 0.838 & 0.710 &       & 500   &       & 0.974 &       & 0.936 & 0.942 & 0.934 & 0.899 & 0.774 \\
      &       &       &       & [1.14] &       & [1.02] & [1.08] & [1.11] & [1.19] & [1.33] &       &       &       & [1.17] &       & [1.14] & [1.17] & [1.21] & [1.27] & [1.45] \\
\midrule
(2,0.1,0.2) &       & 300   &       & 0.840 &       & 0.558 & 0.602 & 0.602 & 0.570 & 0.440 &       & 300   &       & 0.820 &       & 0.634 & 0.669 & 0.670 & 0.632 & 0.507 \\
$\rightarrow$(2,0.1,0.4) &       &       &       & [1.37] &       & [0.95] & [1.05] & [1.10] & [1.16] & [1.16] &       &       &       & [1.33] &       & [1.08] & [1.17] & [1.22] & [1.28] & [1.33] \\
\cmidrule{3-11}\cmidrule{13-21}      &       & 500   &       & 0.972 &       & 0.894 & 0.907 & 0.903 & 0.870 & 0.754 &       & 500   &       & 0.963 &       & 0.926 & 0.942 & 0.944 & 0.923 & 0.819 \\
      &       &       &       & [1.05] &       & [0.97] & [1.00] & [1.01] & [1.01] & [1.06] &       &       &       & [1.04] &       & [1.00] & [1.03] & [1.06] & [1.07] & [1.15] \\
\bottomrule
\end{tabular} \label{tab8}
}
\end{table}%

Next, we consider the cases where data are contaminated by either additive outliers (AO) or innovation outliers (IO). Following the scheme of Fried {\it et al.} (2013), we observe AO-contaminated process $\{X_{o,t}\}$ instead of $\{X_t\}$ in $(\ref{INGARCH})$, such that $X_{o,t}=X_t+p_t X_{c,t}$, where $p_t$'s are i.i.d. Bernoulli random variables with success probability $p$ and $X_{c,t}$'s are  i.i.d. Poisson random
variables with mean $\gamma$. It is assumed that $p_t$, $X_{c,t}$, and $X_t$ are all independent.  IO-contaminated samples are generated by replacing $\lambda_t$ by $\lambda_{o,t}= \lambda_t +p_t \lambda_{c,t}$, where $\lambda_{c,t}$'s are  i.i.d. Poisson random variables with mean $\gamma$ (cf. Fokianos and Fried (2010)). $p_t$, $\lambda_{c,t}$, and $\lambda_t$ are also assumed to be all independent. In both cases, simulations are conducted with $(p,\gamma)$= (0.01,10), (0.01,20), and (0.03,10).
To see the influence of outliers on each test statistics, we define the following ratio:
\[d_\A:= \frac{\text{Empirical size (resp. power) of }T_n^\A\text{ obtained in contaminated case}}{\text{Empirical size (resp. power) of }T_n^\A\text{ obtained in uncontaminated case}}.\]
Here, $d_\A$ with $\A=0$ denotes the corresponding values for the score test $T_n$. In evaluating sizes, figures relatively larger than 1 mean that size distortions are caused by outliers. For empirical powers, values less than 1 indicate power losses by outliers. If the ratio of a test is close to 1, the test can be considered robust against outliers.

Empirical sizes for contaminated cases are presented in Tables \ref{tab3}-\ref{tab5}. It should first be noted that in almost all the cases, $T_n$  yields empirical sizes much larger than the significance level 5\%.  The values of $d_0$ are distributed between 2.05 and 9.13, and the ratio for $T_n$ tends to increase as either $p$ or $\gamma$ increases, indicating that size distortion becomes more severe in such situations. In contrast, each $T_n^\A$ achieves excellent sizes and $d_\A$'s are observed to be close to 1 in most cases considered.  This implies that $T_n^\A$ performs consistently whether outliers exist or not.  Comparing to the results for the empirical sizes, the empirical powers displayed in Tables \ref{tab6}-\ref{tab8} shows that the powers of $T_n$ and $T_n^\A$ are not so sensitive to outliers. Also, significant power losses are not observed. However, one can see that $d_\A$ of $T_n^\A$ are comparatively closer to 1 than that of $T_n$ in most cases, which indicates that $T_n^\A$ is less affected by outliers. Recalling that no size distortions are observed, we can see that the proposed test performs adequately regardless of outliers. In the cases where $d_\A$ is relatively higher than 1, it may be understood that outliers additionally raise the rejection ratios of the null hypothesis. For example, see the case that  $\theta$ changes from (2,0.1,0.2) to (2.5, 0.1,0.2) and $n$=300 in Table \ref{tab7}. $d_0$ for AO and IO contaminations are obtained to be 2.01 and 1.88, respectively.

Overall, our simulation results demonstrate the validity and strong robustness of the proposed test. Furthermore, our test is quite successful when the parameter lies near a boundary. Thus, our test can be a promising tool in testing for parameter change when outliers are suspected to exist.


\section{Conclusion}
In this study, we proposed a robust test for parameter change in Poisson AR models. To construct a test statistics, we used the density power divergence by Basu et al. (1998), and thus our test can be considered as extension of the score test induced from Kullback-Leibler divergence. Under the regularity conditions, we derived the null limiting distribution of the proposed test. The simulation study showed that our test produces excellent sizes and reasonably good powers regardless of the presence of outliers, while the existing score is compromised by outliers.  Therefore, our test can be a useful tool in testing for parameter change when outliers are suspected to contaminate data.

Our test procedure can be applied to other type of inter-valued time series models. For this, asymptotic properties of MDPDE for the models need to be established first. We leave the extension to other models as a task for our future study.\\


\section{Appendix}
In this appendix, we provide the proofs for the Theorem \ref{Thm_Score} in Section 3.

\begin{lemma}\label{Lm1} Suppose that the conditions in Theorem \ref{Thm_Score} hold. Then, under $H_0$, we have
\[\frac{1}{\sqrt{n}}\pa_{\theta}\tilde{H}_{\A,[ns]}(\T_0)-\frac{[ns]}{n}\frac{1}{\sqrt{n}}\,\pa_{\theta}\tilde{H}_{\A,n}(\theta_0)
\stackrel{w}{\longrightarrow}
 (1+\A) K_\A^{1/2}B^o_d(s)\quad \rm{in}\ \ \mathbb{D}\,\big( [0,1],\, \mathbb{R}^d\big).\]
\end{lemma}
\begin{proof}
Note that
\begin{eqnarray*}
\frac{1}{\sqrt{n}} \pa_\T \tilde
H_{\alpha,[ns]}(\theta_0)\stackrel {w}{\longrightarrow} (1+\A) K_\A^{1/2}B_d(s)~~in~~\mathbb {D}([0,1],\mathbb {R}^d )
\end{eqnarray*}
where ${B}_d$ is a standard $d$-dimensional Brownian motion (
see Lemma 1 in Kang and Song (2015)), which consequently yield the lemma.
\end{proof}


\begin{lemma}\label{Lm2} Suppose that the conditions in Theorem \ref{Thm_Score} hold. Then, under $H_0$, we have
\begin{eqnarray*}
\frac{1}{n}\sum_{t=1}^{n}\sup_{\theta \in\Theta}\left\|\paa  l_{\A,t}(\theta)
-\paa  \tilde l_{\A,t}(\theta)\right\|&=& o(1)\quad a.s.
\end{eqnarray*}
\end{lemma}
\begin{proof}
The proof for Lemma 7 in Kang and Lee (2014b) include the stated results. Thus, we omit the proof.
\end{proof}


\begin{lemma}\label{Lm3} Suppose that the conditions in Theorem \ref{Thm_Score} hold. Then, under $H_0$, we have that for any $\ep>0$, there exists a neighborhood $\mathcal{N}_\ep$ of $\T_0$ such that
\begin{eqnarray*}\label{K}
\E\sup_{\theta \in \mathcal{N}_\ep}\left\|\paa  l_{\A,t}(\T)-\paa  l_{\A,t}(\T_0)\right\|<\ep.
\end{eqnarray*}
\end{lemma}
\begin{proof}
 From Lemma 5 of Kang and Lee (2014b), we have
\[\E \sup_{\T \in \Theta} \big\|\paa l_{\A,t}(\T)-\paa l_{\A,t}(\T_0)\big\| <\infty.\]
Let $\mathcal{N}_{1/n}(\T_0)=\{\T\in \Theta|\,\|\T-\T_0\|\leq 1/n\}$. Since $\paa l_{\A,t}(\T)$ is continuous in $\T$, it follows from the bounded convergence theorem that
\[ \E \sup_{\T \in \mathcal{N}_{1/n}(\T_0)} \big\|\paa l_{\A,t}(\T)-\paa l_{\A,t}(\T_0)\big\|=o(1)\quad \text{as}\quad n\rightarrow \infty. \]
This establishes the lemma.
\end{proof}

\begin{lemma}\label{Lm4}Suppose that the conditions in Theorem \ref{Thm_Score} hold. Then, under $H_0$, we have
\begin{eqnarray*}
\sup_{0\leq s \leq1}\Big\|\frac{1}{n}\paa \tilde{H}_{\A,[ns]}(\T^*_{\A,n,s})\frac{1}{1+\A}J_\A^{-1}\,\frac{1}{\sqrt{n}}\,\pa_{\theta}\tilde{H}_{\A,n}(\T_0)
+\frac{[ns]}{n}\frac{1}{\sqrt{n}}\,\pa_{\theta}\tilde{H}_{\A,n}(\T_0)\Big\|=o_P(1).
\end{eqnarray*}
\end{lemma}
\begin{proof}
Without confusion, we shall denote $\theta^*_{n,k/n}$ by $\theta^*_{n,k}$ for $k \leq n$. Note that
\begin{eqnarray*}
&&\sup_{0\leq s \leq1}\Big\|\frac{1}{n}\paa \tilde{H}_{\A,[ns]}(\T^*_{\A,n,s})\frac{1}{1+\A}J_\A^{-1}\,\frac{1}{\sqrt{n}}\,\pa_{\theta}\tilde{H}_{\A,n}(\T_0)
+\frac{[ns]}{n}\frac{1}{\sqrt{n}}\,\pa_{\theta}\tilde{H}_{\A,n}(\T_0)\Big\|\\
&&\leq \Big\|\frac{1}{1+\A}J_\A^{-1} \frac{1}{\sqrt{n}}\,\pa_{\theta}\tilde{H}_{\A,n}(\T_0) \Big\|\,\max_{ 1\leq k\leq n} \frac{k}{n}\Big\|\frac{1}{k}\paa \tilde{H}_{\A,k}(\T^*_{\A,n,k})+(1+\A)J_\A\Big\|.
\end{eqnarray*}
Owing to Lemma \ref{Lm1}, it suffices to show that
\begin{eqnarray}\label{Lm4.1}
\max_{ 1\leq k\leq n} \frac{k}{n}\Big\|\frac{1}{k}\paa \tilde{H}_{\A,k}(\T^*_{\A,n,k})+(1+\A)J_\A\Big\|=o_P(1).
\end{eqnarray}
For any $\ep>0$, observe that we can take a neighborhood $\mathcal{N}_\ep(\T_0)=\{\theta\in\Theta|\, \|\theta-\theta_0\|< r_\ep\}$ such that
\begin{eqnarray}\label{Lm4.2}
\E\sup_{\theta \in \mathcal{N}_\ep(\T_0)}\| \paa l_{\A,t} (\theta)- \paa l_{\A,t} (\theta_0)\| < \ep,
\end{eqnarray}
due to Lemma \ref{Lm3}.
Since $\HT$ converges almost surely to $\theta_0$, we have that for sufficiently large $n$,
\begin{eqnarray*}
&&\max_{ 1\leq k\leq n} \frac{k}{n}\Big\|\frac{1}{k}\paa \tilde{H}_{\A,k}(\T^*_{\A,n,k})+(1+\A)J_\A\Big\|\\
&&\leq \max_{ 1\leq k\leq n} \frac{1}{n}\big\|\paa \tilde{H}_{\A,k}(\T^*_{\A,n,k})-\paa H_{\A,k}(\T^*_{\A,n,k})\big\|
+ \max_{ 1\leq k\leq n} \frac{1}{n}\big\|\paa H_{\A,k}(\T^*_{\A,n,k})-\paa H_{\A,k}(\T_0)\big\|\\
&&\hspace{0.5cm}
+ \max_{ 1\leq k\leq n} \frac{k}{n}\Big\|\frac{1}{k}\paa H_{\A,k}(\T_0)+(1+\A)J_\A\Big\|\\
&&\leq \frac{1}{n} \sum_{t=1}^n \sup_{\T\in\Theta} \| \paa\,\tilde{l}_{\A,t}(\T)-\paa\,l_{\A,t}(\T)\|
+\frac{1}{n} \sum_{t=1}^n \sup_{\theta\in \mathcal{N}_{\ep}(\theta_0)} \| \paa\,l_{\A,t}(\T)-\paa\,l_{\A,t}(\T_0)\|\\
&&\hspace{0.5cm}
+ \max_{ 1\leq k\leq n} \frac{k}{n}\Big\|\frac{1}{k}\paa H_{\A,k}(\T_0)+(1+\A)J_\A\Big\|\\
&&:=A_n +B_n+ C_n\quad a.s.
\end{eqnarray*}
First, one can see that  $A_n=o(1)$ a.s. by Lemma \ref{Lm2}. Also, using (\ref{Lm4.2}) and the fact that $\{\paa l_{\A,t} (\theta)\}$ is stationary and ergodic, we have
\begin{eqnarray*}
\lim_{n\rightarrow\infty}B_n =
\E\sup_{\theta \in \mathcal{N}_\ep(\theta_0)}\| \paa l_{\A,t} (\T)- \paa l_{\A,t} (\T_0)\| < \ep\quad a.s.
\end{eqnarray*}
Furthermore, noting that $\|\paa H_{\A,n}(\T_0)/n+(1+\A)J_\A\|$ converges to zero almost surely, it can be shown that
\begin{eqnarray*}
&&\max_{1\leq k \leq \sqrt{n}}  \frac{k}{n}\Big\|\frac{1}{k}\paa H_{\A,k}(\T_0)+(1+\A)J_\A\Big\|
\leq \frac{1}{\sqrt{n}} \sup_n \Big\|\frac{1}{n}\paa H_{\A,n}(\T_0)+(1+\A)J_\A\Big\|=o(1)\quad a.s.
\end{eqnarray*}
and
\begin{eqnarray*}
&&\max_{\sqrt{n} < k \leq n} \Big\|\frac{1}{k}\paa H_{\A,k}(\T_0)+(1+\A)J_\A\Big\| =o(1)\quad a.s.,
\end{eqnarray*}
which subsequently yield $C_n=o(1)$ a.s. and hence (\ref{Lm4.1}) is yielded. This completes the proof.
\end{proof}

\begin{lemma}\label{Lm5}Suppose that the conditions in Theorem \ref{Thm_Score} hold. Then, under $H_0$, we have
\begin{eqnarray*}
\sup_{0\leq s\leq1} \frac{1}{n} \big\|\paa \tilde{H}_{\A,[ns]}(\theta^*_{\A,n,s}) \big(B_{\A,n} +(1+\A)J_\A\big)\sqrt{n}(\HT-\T_0)\big\|=o_P(1).
\end{eqnarray*}
\end{lemma}
\begin{proof}
Due to (\ref{Lm4.1}), we have
\[\|B_{\A,n} +(1+\A)J_\A\| \leq \max_{ 1\leq k\leq n} \frac{k}{n}\Big\| \frac{1}{k}\paa \tilde{H}_{\A,k}(\T^*_{\A,n,k})+(1+\A)J_\A\Big\|=o_P(1)\]
and
\[\sup_{0\leq s\leq1} \frac{1}{n} \|\paa \tilde{H}_{\A,[ns]}(\T^*_{\A,n,s}) \|
\leq\max_{ 1\leq k\leq n} \frac{1}{n} \|\paa \tilde{H}_{\A,k}(\T^*_{\A,n,k})+(1+\A)J_\A\|
+\|(1+\A)J_\A\|=O_P(1),\]
together with the fact that $\sqrt{n}(\HT-\T_0)=O_P(1)$,
the lemma is established.
\end{proof}


\vspace{1cm}
\noindent{\bf References}
 \begin{description}
\item Aue, A. and Horváth, L. (2013). Structural breaks in time series. {\it Journal of Time Series Analysis} {\bf 34}, 1–16.
\item Basu, A., Harris, I. R., Hjort, N. L. and Jones, M. C. (1998). Robust and efficient estimation by minimizing a density power divergence. {\it Biometrika} {\bf 85}, 549-559.
\item Basu, A., Mandal, A., Martin, N. and Pardo, L. (2016). Generalized Wald-type tests based on minimum density power divergence estimators. {\it Statistics} {\bf 50(1)}, 1-26.
\item Batsidis, A., Horváth, L., Martín, N., Pardo, L. and Zografos, K. (2013). Change-point detection in multinomial data using phi-divergence test statistics. {\it Journal of Multivariate Analysis} {\bf 118}, 53-66.
\item Diop, M. L. and Kengne, W. (2017) Testing parameter change in general integer‐valued time series. {\it Journal of Time Series Analysis} {\bf 38}, 880-894.
\item Doukhan, P., Fokianos, K. and Tj{\o}stheim, D. (2012). On weak dependence conditions for Poisson autoregressions. {\it Statistics and Probability Letters} {\bf 82}, 942-948.
\item Doukhan, P. and Kengne, W. (2015). Inference and testing for structural change in general Poisson autoregressive models. {\it Electronic Journal of Statistics} {\bf 9}, 1267-1314.
\item Ferland, R., Latour, A. and Oraichi, D. (2006). Integer-valued GARCH process. {\it Journal of Time Series Analysis } {\bf 27}, 923-942.
\item Fokianos, K. and Fried, R. (2010). Interventions in INGARCH processes. {\it Journal of Time Series Analysis } {\bf 31}, 210-225.
\item Fokianos, K., Rahbek, A. and Tj{\o}stheim, D. (2009). Poisson autoregression. {\it Journal of the American Statistical Association } {\bf 104}, 1430-1439.
\item Fried, R., Agueusop, I., Bornkamp, B., Fokianos, K., Fruth, J. and Ickstadt, K. (2013). Retrospective Bayesian outlier detection in INGARCH series. {\it Statistics and Computing} {\bf 25(2)}, 365-374.
\item Horv\'{a}th, L. and Parzen, E. (1994). Limit theorems for Fisher-score change processes. {\it Lecture Notes-Monograph Series} {\bf 23}, 157-169.
\item Horváth, L. and Rice, G. (2014). Extensions of some classical methods in change point analysis. {\it TEST} {\bf 23}, 219-255.
\item Hudecov\'{a}, \v{S}., Hu\v{s}kov\'{a}, M. and Meintanis, S. G. (2017). Tests for structural changes in time series of counts. {\it Scandinavian Journal of Statistics} {\bf 44}, 843-865.
\item Jung, R. C. and Tremayne, A. R. (2011). Useful models for time series of counts or simply wrong ones? {\it AStA Advances in Statistical Analysis } {\bf 95}, 59-91.
\item Kang, J. and Lee, S. (2009). Parameter change test for random coefficient integer‐valued autoregressive processes with application to polio data analysis. {\it Journal of Time Series Analysis} {\bf 30(2)}, 239-258.
\item Kang, J. and Lee, S. (2014a). Parameter change test for Poisson autoregressive models. {\it Scandinavian Journal of Statistics} {\bf 41(4)}, 1136-1152.
\item Kang, J. and Lee, S. (2014b). Minimum density power divergence estimator for Poisson autoregressive models. {\it Computational Statistics and Data Analysis} {\bf 80}, 44-56.
\item Kang, J. and Song, J. (2015). Robust parameter change test for Poisson autoregressive models. {\it Statistics and Probability Letters} {\bf 104}, 14-21.
\item Kang, J. and Song, J. (2017). Score test for parameter change in Poisson autoregressive models. {\it Economics Letters} {\bf 160}, 33-37.
\item Knoblauch, J., Jewson, J. and Damoulas, T. (2018).Doubly robust Bayesian inference for non-stationary streaming data with $\beta$-divergences. {\it Advances in Neural Information Processing Systems 31}.
\item Lee, S. and Na, O. (2005). Test for parameter change based on the estimator minimizing density-based divergence measures. {\it Annals of the Institute of Statistical Mathematics} {\bf 57(3)}, 553-573.
\item Neumann, M. (2011). Absolute regularity and ergodicity of Poisson count processes. {\it Bernoulli} {\bf 17}, 1268-1284.
\item Pardo, L. (2006). {\it Statistical Inference Based on Divergence Measures}, Chapman and Hall/CRC.
\item Song, J. and Kang, J. (2019). Test for parameter change in the presence of outliers: the density power divergence based approach. arXiv:1907.00004.
\item Weiß, C. H. (2009). {\it Categorical time series analysis and applications in statistical quality control.} dissertation. de-Verlag im Internet GmbH.
\item Weiß, C. H. (2010). The INARCH (1) model for overdispersed time series of counts. {\it Communications in Statistics-Simulation and Computation} {\bf 39(6)}, 1269-1291.
\item Zeger, S. L. (1988). A regression model for time series of counts. {\it Biometrika } {\bf 75}, 621-629.
\item Zhu, R. and Joe, H. (2006). Modelling count data time series with Markov processes based on binomial thinning. {\it Journal of Time Series Analysis} {\bf 27(5)}, 725-738.

\end{description}

\newpage

\end{document}